\documentclass[conference]{IEEEtran}
\usepackage[utf8]{inputenc}
\usepackage{amsmath}
\usepackage{amsfonts}
\usepackage{amsthm}
\usepackage{amssymb}
\usepackage{amsmath}
\usepackage{amsfonts}
\usepackage{amsthm}
\usepackage{amssymb}
\usepackage{tikz}

\newtheorem{thm}{Theorem}

\newtheorem{cor}{Corollary}

\newtheorem{ass}{Assumption}

\DeclareMathOperator*{\argmin}{argmin}

\long\def\symbolfootnote[#1]#2{\begingroup%
	\def\thefootnote{\fnsymbol{footnote}}\footnote[#1]{#2}\endgroup} 
\IEEEoverridecommandlockouts

\begin{document}

\title{Improved Information Theoretic Generalization Bounds for Distributed and Federated Learning}
\author{L. P. Barnes$^*$, Alex Dytso$^\dagger$, and H. V. Poor$^*$ \\
$^*$Princeton University, Department of Electrical and Computer Engineering \\
 $^\dagger$New Jersey Institute of Technology, Department of Electrical and Computer Engineering
\thanks{This work was supported in part by the National Science Foundation under Grant CCF-1908308.}
}

\maketitle

\begin{abstract}
We consider information-theoretic  bounds on expected generalization error for statistical learning problems in a networked setting. In this setting, there are $K$ nodes, each with its own independent dataset, and the models from each node have to be aggregated into a final centralized model. We consider both simple averaging of the models as well as more complicated multi-round algorithms. We give upper bounds on the expected generalization error for a variety of problems, such as those with Bregman divergence or Lipschitz continuous losses, that demonstrate an improved dependence of $1/K$ on the number of nodes. These ``per node'' bounds are in terms of the mutual information between the training dataset and the trained weights at each node, and are therefore useful in describing the generalization properties inherent to having communication or privacy constraints at each node.
\end{abstract}

\section{Introduction}
A key property of machine learning systems is their ability to generalize to new and unknown data. Such a system is trained on a particular set of data, but must then perform well even on new datapoints that have not previously been considered. This ability, deemed generalization, can be formulated in the language of statistical learning theory by considering the generalization error of an algorithm, i.e, the difference between the population risk of a model trained on a particular dataset and the empirical risk for the same model and dataset. We say that a model generalizes well if it has a small generalization error, and because models are often trained by minimizing empirical risk or some regularized version of it, a small generalization error also implies a small population risk which is the average loss over new samples taken randomly from the population. It is therefore of interest to upper bound generalization error and understand which quantities control it, so that we can quantify the generalization properties of a machine learning system and offer guarantees about how well it will perform.

In recent years, it has been shown that information theoretic quantities such as mutual information can be used to bound generalization error under assumptions on the tail of the distribution of the loss function \cite{russo_zou,xu_raginsky,veeravalli}. In particular, when the loss function is sub-Gaussian, the expected generalization error can scale at most with the square root of the mutual information between the training dataset and the model weights \cite{xu_raginsky}. These bounds offer an intuitive explanation for generalization and overfitting -- if an algorithm uses only limited information from its training data, then this will bound the expected generalization error and prevent overfitting. Conversely, if a training algorithm uses all of the information from its training data in the sense that the model is a deterministic function of the training data, then this mutual information can be infinite and there is the possibility of unbounded generalization error and thus overfitting.

Another modern focus of machine learning systems has been that of distributed and federated learning \cite{federated0,federated1,federated2}. In these systems, data is generated and processed in a distributed network of machines. The main differences between the distributed and centralized settings are the information constraints imposed by the network. There has been considerable interest in understanding the impact of both communication constraints \cite{deepgrad,rtopk} and privacy constraints \cite{warner,dwork1,whatcanwelearn,cuff} on the performance of machine learning systems, and in designing protocols that efficiently train systems under these constraints.

Since both communication and local differential privacy constraints can be thought of as special cases of mutual information constraints, they should pair naturally with some form of information theoretic generalization bound in order to induce control over the generalization error of the distributed machine learning system. The information constraints inherent to the network can themselves give rise to tighter bounds on generalization error and thus provide better guarantees against overfitting. Along these lines, in recent work \cite{spawc}, a subset of the present authors introduced the framework of using information theoretic quantities to bound both expected generalization error and a measure of privacy leakage in distributed and federated learning systems. The generalization bounds in this work, however, are essentially the same as those obtained by thinking of the entire system, from the data at each node in the network to the final aggregated model, as a single centralized algorithm. Any improved generalization guarantees from these bounds would remain implicit in the mutual information terms involved.

In this work, we develop improved bounds on expected generalization error for distributed and federated learning systems. Instead of leaving the differences between these systems and their centralized counterparts implicit in the mutual information terms, we bring analysis of the structure of the systems directly into the bounds. By working with the contribution from each node separately, we are able to derive upper bounds on expected generalization error that scale with the number of nodes $K$ as $O\left( \frac{1}{K} \right)$ instead of $O\left( \frac{1}{\sqrt{K}}\right)$. This improvement is shown to be tight for certain examples such as learning the mean of a Gaussian with squared $\ell^2$ loss. We develop bounds that apply to distributed systems in which the submodels from each one of $K$ different nodes are averaged together, as well as bounds that apply to more complicated multiround stochastic gradient descent (SGD) algorithms such as in federated learning. For linear models with Bregman divergence losses, these ``per node'' bounds are in terms of the mutual information between the training dataset and the trained weights at each node, and are therefore useful in describing the generalization properties inherent to having communication or privacy constraints at each node. For arbitrary nonlinear models that have Lipschitz continuous losses, the improved dependence of $O\left( \frac{1}{K} \right)$ can still be recovered, but without a description in terms of mutual information. We demonstrate the improvements given by our bounds over the existing information theoretic generalization bounds via simulation of a distributed linear regression example.

\subsection{Technical Preliminaries}
Suppose we have independent and identically distributed (i.i.d.) data $Z_i \sim \pi$ for $i=1,\ldots,n$ and let $S = (Z_1,\ldots,Z_n)$. Suppose further that $W = \mathcal{A}(S)$ is the output of a potentially stochastic algorithm. Let $\ell(W,Z)$ be a real-valued loss function and define
$$L(w) = \mathbb{E}_\pi[\ell(w,Z)]$$
to be the population risk for weights (or model) $w$. We similarly define
$$L_s(w) = \frac{1}{n} \sum_{i=1}^n \ell(w,z_i)$$
to be the empirical risk on dataset $s$ for model $w$. The generalization error for dataset $s$ is then
$$\Delta_\mathcal{A}(s) = L(\mathcal{A}(s)) - L_s(\mathcal{A}(s))$$
and the expected generalization error is
\begin{equation} \label{eq:gen_error}
\mathbb{E}_{S\sim \pi^n} [\Delta_\mathcal{A}(S)] = \mathbb{E}_{S\sim \pi^n} [ L(\mathcal{A}(S)) - L_S(\mathcal{A}(S))]
\end{equation}
where the expectation is also over any randomness in the algorithm. Below we present some standard results on the expected generalization error that will be needed.

\begin{thm}[Leave-one-out Expansion -- Lemma 11 in \cite{SSSS}] \label{thm:leave}
Let $S^{(i)} = (Z_1,\ldots,Z_i',\ldots,Z_n)$ be a version of $S$ with $Z_i$ replaced by an i.i.d. copy $Z_i'$. Denote $S' = (Z'_1,\ldots,Z'_n)$. Then
$$\mathbb{E}_{S\sim \pi^n} [\Delta_\mathcal{A}(S)] = \frac{1}{n}\sum_{i=1}^n\mathbb{E}_{S,S'}[ \ell(\mathcal{A}(S),Z_i') - \ell(\mathcal{A}(S^{(i)}),Z_i')] \; .$$
\end{thm}
\begin{proof}
Observe that
\begin{equation} \label{eq:pop_risk}
\mathbb{E}_{S\sim \pi^n} [ L(\mathcal{A}(S))] =  \mathbb{E}_{S,S'}[ \ell(\mathcal{A}(S),Z_i')]
\end{equation}
for each $i$ and that
\begin{align} \label{eq:emp_risk}
\mathbb{E}_{S\sim \pi^n} [L_S(\mathcal{A}(S))] & = \frac{1}{n} \sum_{i=1}^n\mathbb{E}_{S\sim \pi^n} \left[ \ell(\mathcal{A}(S),Z_i)\right] \nonumber \\
& = \frac{1}{n} \sum_{i=1}^n \mathbb{E}_{S,S'\sim \pi^n} \left[ \ell(\mathcal{A}(S^{(i)}),Z'_i)\right] \; .
\end{align}
Putting \eqref{eq:pop_risk} and \eqref{eq:emp_risk} together with \eqref{eq:gen_error} yields the result.
\end{proof}
In many of the results in this paper, we will use one of the two following assumptions.
\begin{ass} \label{ass1}
The loss function $\ell(\widetilde W, \widetilde Z)$ satisfies
$$\log\mathbb{E}\left[\exp\left(\lambda\left(\ell(\widetilde W, \widetilde Z)-\mathbb{E}[\ell(\widetilde W, \widetilde Z)]\right)
\right)\right] \leq \psi(-\lambda) $$ for $\lambda \in (b,0]$, $\psi(0)=\psi'(0)=0$, where $\widetilde W, \widetilde Z$ are taken independently from the marginals for $W,Z$, respectively,
\end{ass} 
\noindent The next assumption is a special case of the previous one with $\psi(\lambda) = \frac{R^2\lambda^2}{2} \; .$
\begin{ass} \label{ass2}
The loss function $\ell(\widetilde W, \widetilde Z)$ is sub-Gaussian with parameter $R^2$ in the sense that
$$\log\mathbb{E}\left[\exp\left(\lambda\left(\ell(\widetilde W, \widetilde Z)-\mathbb{E}[\ell(\widetilde W, \widetilde Z)]\right)\right)\right] \leq  \frac{R^2\lambda^2}{2} \; .$$
\end{ass}

\begin{thm}[Theorem 2 in \cite{veeravalli}]\label{thm:info_theoretic}
Under Assumption \ref{ass1},
$$\mathbb{E}_{S\sim\pi^n}[\Delta_\mathcal{A}(S)] \leq \frac{1}{n}\sum_{i=1}^n \psi^{*-1}(I(W;Z_i))$$
where
$$\psi^{*-1}(y) = \inf_{\lambda\in[0,b)}\left(\frac{y+\psi(\lambda)}{\lambda}\right) \; .$$
\end{thm}

Recall that for a continuously differentiable and strictly convex function $F:\mathbb{R}^m\to\mathbb{R}$, we define the associated Bregman divergence \cite{bregman1967relaxation} between two points $p,q\in\mathbb{R}^m$ to be
$$D_F(p,q) = F(p) - F(q) - \langle\nabla F(q),p-q\rangle \; ,$$
where $\langle \cdot, \cdot \rangle$ denotes the usual inner product.
\section{Distributed Learning and Model Aggregation}
Now suppose that there are $K$ nodes each with $n$ samples. Each node $k=1,\ldots,K$ has dataset $S_k = (Z_{1,k},\ldots,Z_{n,k})$ with $Z_{i,k}$ taken i.i.d. from $\pi$. We use $S=(S_1,\ldots,S_K)$ to denote the entire dataset of size $nK$. Each node locally trains a model $W_k = \mathcal{A}_k(S_k)$ with algorithm $\mathcal{A}_k$. After each node locally trains its model, the models $W_k$ are then combined to form the final model $\widehat{W}$ using an aggregation algorithm $\widehat W = \widehat{\mathcal{A}}(W_1,\ldots,W_K)$. See Figure \ref{fig1b}. In this section we will assume that $W_k\in\mathbb{R}^d$ and that the aggregation is done by simple averaging, i.e.,
$$\widehat W = \frac{1}{K}\sum_{k=1}^K W_k \; .$$
Define $\mathcal{A}$ to be the total algorithm from data $S$ to the final weights $\widehat W$ so that $\widehat W = \mathcal{A}(S) \; .$

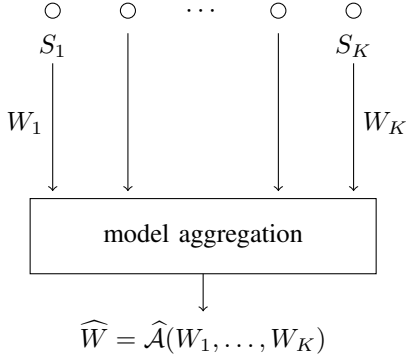
\begin{figure}
\centering{
\begin{tikzpicture}
\draw (0,0) circle (0.1cm); \node [below] at (0,-0.2) {$S_1$}; 
\draw (1,0) circle (0.1cm); 
\node at (2,0) {$\cdots$}; 
\draw (3,0) circle (0.1cm);
\draw (4,0) circle (0.1cm); \node [below] at (4,-0.2) {$S_K$}; 
\draw [->] (0, -0.7) -- (0, -2.4); \draw [->] (1, -0.3) -- (1, -2.4); 
\draw [->] (3, -0.3) -- (3, -2.4); \draw [->] (4, -0.7) -- (4, -2.4);
\draw (-0.3,-3.5) rectangle (4.3,-2.5); 
\node at (2,-3) {model aggregation}; 
\draw [->] (2,-3.5) -- (2, -4); \node [below] at (2,-4) {$\widehat W = \widehat{\mathcal{A}}(W_1,\ldots,W_K)$}; 
\node [left] at (0, -1.5) {$W_1$}; \node [right] at (4,-1.5) {$W_K$}; 
\end{tikzpicture}
}
\caption{The distributed learning setting with model aggregation.}
\label{fig1b}
\end{figure}

\begin{thm}\label{thm:convex}
Suppose that $\ell(\cdot,z)$ is a convex function of $w\in\mathbb{R}^d$ for each $z$ and that $\mathcal{A}_k$ represents the empirical risk minimization algorithm on local dataset $S_k$ in the sense that
$$W_k = \mathcal{A}_k(S_k) = \argmin_w \sum_{i=1}^n \ell(w,Z_{i,k}) \; .$$
Then
$$\Delta_\mathcal{A}(s) \leq \frac{1}{K}\sum_{k=1}^K \Delta_{\mathcal{A}_k}(s_k) \; .$$ 
\end{thm}
\begin{proof}
\begin{align}
 \Delta_\mathcal{A} & (s)  =  \mathbb{E}_{Z\sim\pi}[\ell(\mathcal{A}(s),Z)] - \frac{1}{nK} \sum_{i,k} \ell(\mathcal{A}(s),z_{i,k}) \nonumber \\
& = \mathbb{E}_{Z\sim\pi}\left[\ell\left(\frac{1}{K}\sum_{k=1}^K w_k,Z\right)\right] - \frac{1}{nK} \sum_{i,k} \ell(\mathcal{A}(s),z_{i,k}) \nonumber \\
& \leq \frac{1}{K}\sum_{k=1}^K\mathbb{E}_{Z\sim\pi}[\ell(w_k,Z)] - \frac{1}{nK} \sum_{i,k} \ell(\mathcal{A}(s),z_{i,k}) \label{eq:convex_loss_1}\\
& \leq \frac{1}{K}\sum_{k=1}^K\mathbb{E}_{Z\sim\pi}[\ell(w_k,Z)] - \frac{1}{K}\sum_{k=1}^K\min_w\frac{1}{n}\sum_{i=1}^n \ell(w,z_{i,k}) \label{eq:convex_loss_2}\\
& = \frac{1}{K}\sum_{k=1}^K \Delta_{\mathcal{A}_k}(s_k). \nonumber
\end{align}
In the above display, line \eqref{eq:convex_loss_1} follows by the convexity of $\ell$ via Jensen's inequality, and line \eqref{eq:convex_loss_2} follows by minimizing the empirical risk over each node's local dataset, which exactly corresponds to what each node's local algorithm $\mathcal{A}_k$ does.
\end{proof}
While Theorem \ref{thm:convex} seems like a nice characterization of generalization bounds for the aggregate model -- in that the aggregate generalization error cannot be any larger than the average generalization errors over each node -- it does not offer any improvement in the expected generalization error that one might expect given $nK$ total samples instead of just $n$ samples. A naive application of the information theoretic generalization bounds from Theorem \ref{thm:info_theoretic}, followed by the data processing inequality $I(\widehat W;Z_{i,k}) \leq I(W_k;Z_{i,k})$, runs into the same problem.
\subsection{Improved Bounds}
In this subsection, we prove bounds on expected generalization error that remedy the above shortcomings. In particular, we would like the following two properties.
\begin{itemize}
\item[(a)] The bound should decay with the number of nodes $K$ in order to take advantage of the total dataset from all $K$ nodes.
\item[(b)] The bound should be in terms of the information theoretic quantities $I(W_k;S_k)$ which can represent (or be upper bounded by) the capacities of the channels that the nodes are communicating over. This can, for example, represent a communication or local differential privacy constraint for each node.
\end{itemize}
At a high level, we will improve on the bound from Theorem \ref{thm:convex} by taking into account the fact that a small change in $S_k$ will only change $\widehat W$ by a fraction  $\frac{1}{K}$ of  the amount that it will change $W_k$.
In the case that $W$ is a linear or location model, and the loss $\ell$ is a Bregman divergence, we can obtain an upper bound on expected generalization error that satisfies both properties (a) and (b) as follows.

\begin{ass} \label{new_assumption}
When $Z=(X,Y)$ are labeled pairs and for loss functions of type (i) in Theorem \ref{thm:breg} below, we assume that
\begin{align*}
&\mathbb{E}_{S,S'}[F(\langle X'_{i,k}, \mathcal{A}(S)\rangle) - F(\langle X'_{i,k}, \mathcal{A}(S^{(i,k)})\rangle)] \\
&\leq \frac{1}{K}\left(\mathbb{E}_{S,S'}[F(\langle X'_{i,k}, \mathcal{A}_k(S_k) \rangle) - F(\langle X'_{i,k}, \mathcal{A}_k(S_k^{(i)})\rangle )]\right) \; .
\end{align*}
\end{ass}

Whether or not this assumption holds true will depend on the distributions involved, the training algorithms, and the function $F$. For least squares regression examples similar to those discussed in the last section of this paper, we have verified, through Monte Carlo simulation, that this assumption appears to hold for all parameter values that we tested. It remains an interesting open problem to understand when this holds true.

\begin{thm}[Linear or Location Models with Bregman Loss] \label{thm:breg}
Suppose that Assumption \ref{ass1} holds for each node. Consider the following two cases:
\begin{itemize}
\item[(i)] $\ell(w,(x,y)) = D_F(\langle x,w\rangle,y)$ (with Assumption \ref{new_assumption}),
\end{itemize}
then 
\begin{align*}
\mathbb{E}_{S\sim \pi^{nK}} [\Delta_\mathcal{A}(S)] & \leq \frac{1}{nK^2}\sum_{i,k} \psi^{*-1}\left(I(W_k; Z_{i,k})\right) \\
& \leq \frac{1}{K^2}\sum_{k=1}^K \psi^{*-1}\left(\frac{I(W_k; S_k)}{n}\right) \; .
\end{align*}
\begin{itemize}
\item[(ii)] $\ell(w,z) = D_F(w,z) \; ,$
\end{itemize}
then
$$\mathbb{E}_{S\sim \pi^{nK}} [\Delta_\mathcal{A}(S)] = \frac{1}{K^2}\sum_{k=1}^K \mathbb{E}_{S_k\sim \pi^n}[\Delta_{\mathcal{A}_k}(S_k)]$$
and
\begin{align*}
\mathbb{E}_{S\sim \pi^{nK}} [\Delta_\mathcal{A}(S)] & \leq \frac{1}{nK^2}\sum_{i,k} \psi^{*-1}\left(I(W_k; Z_{i,k})\right) \\
& \leq \frac{1}{K^2}\sum_{k=1}^K \psi^{*-1}\left(\frac{I(W_k; S_k)}{n}\right) \; .
\end{align*}
\end{thm}
\begin{proof}
Here we restrict our attention to case (ii), but the two cases have nearly identical proofs. Using Theorem \ref{thm:leave},
\begin{align} 
& \mathbb{E}_{S\sim \pi^{nK}} [\Delta_\mathcal{A}(S)] \nonumber \\
& = \frac{1}{nK}\sum_{i,k}\mathbb{E}_{S,S'}\left[ \ell(\mathcal{A}(S),Z_{i,k}') - \ell(\mathcal{A}(S^{(i,k)}),Z_{i,k}')\right] \nonumber \\
& \leq  \frac{1}{nK}\sum_{i,k}\mathbb{E}_{S,S'} \bigg[F(\mathcal{A}(S))-F(Z_{i,k}') \nonumber \\
& \quad \quad \quad \quad \quad  \quad \quad - \big\langle\nabla F(Z_{i,k}'),\mathcal{A}(S)-Z_{i,k}' \big\rangle \nonumber \\
& \quad \quad \quad \quad \quad  \quad \quad - F(\mathcal{A}(S^{(i,k)})) + F(Z_{i,k}') \nonumber \\
& \quad \quad \quad \quad \quad \quad \quad + \big\langle\nabla F(Z_{i,k}'),\mathcal{A}(S^{(i,k)})-Z_{i,k}' \big\rangle\bigg] \nonumber \\
 = & \frac{1}{nK}\sum_{i,k}\mathbb{E}_{S,S'} \bigg[\big\langle\nabla F(Z_{i,k}'),\mathcal{A}(S^{(i,k)})-\mathcal{A}(S) \big\rangle\bigg] \label{eq:key_step1} \\
= & \frac{1}{nK^2}\sum_{i,k}\mathbb{E}_{S,S'} \bigg[\big\langle\nabla F(Z_{i,k}'),W_k^{(i)} - W_k \big\rangle\bigg] \; .  \label{eq:key_step2}
\end{align}
In \eqref{eq:key_step2}, we use $W_k^{(i)}$ to denote $\mathcal{A}_k(S_k^{(i)})$. Line \eqref{eq:key_step1} follows by the linearity of the inner product and by canceling the higher order terms $F(\mathcal{A}(S))$ and $ F(\mathcal{A}(S^{(i,k)}))$ which have the same expected values. The key step \eqref{eq:key_step2} then follows by noting that $\mathcal{A}(S^{(i,k)})$ only differs from $\mathcal{A}(S)$ in the submodel coming from node $k$, which is multiplied by a factor of $\frac{1}{K}$ when averaging all of the submodels. By backing out step $\eqref{eq:key_step1}$ and re-adding the appropriate canceled terms we get
$$\mathbb{E}_{S\sim \pi^{nK}} [\Delta_\mathcal{A}(S)] = \frac{1}{K^2}\sum_{k=1}^K \mathbb{E}_{S_k\sim \pi^n}[\Delta_{\mathcal{A}_k}(S_k)] \; .$$
By applying Theorem \ref{thm:info_theoretic},
$$\mathbb{E}_{S\sim \pi^{nK}} [\Delta_\mathcal{A}(S)] \leq \frac{1}{nK^2}\sum_{i,k} \psi^{*-1}\left(I(W_k; Z_{i,k})\right) \;.$$
Then, by noting that $\psi^{*{-1}}$ is non-decreasing and concave,
\begin{align}
\frac{1}{nK^2}\sum_{i,k} & \psi^{*-1}\left(I(W_k; Z_{i,k})\right) \nonumber \\
& \leq \frac{1}{K^2}\sum_{k=1}^K \psi^{*-1}\left(\sum_{i=1}^n\frac{I(W_k; Z_{i,k})}{n}\right) \nonumber \; .
\end{align}
And using
\begin{align*}
\sum_{i=1}^n & I(W_k; Z_{i,k}) \\
& =  \sum_{i=1}^n H(Z_{i,k}) - H(Z_{i,k}|W_k) \\
& \leq \sum_{i=1}^n H(Z_{i,k}|Z_{i-1,k},\ldots,Z_{1,k}) \\
& \quad \quad \quad \quad - H(Z_{i,k}|Z_{i-1,k},\ldots,Z_{1,k},W_k) \\
& = I(W_k; S_k)
\end{align*}
we have
\begin{align*}
\frac{1}{K^2}\sum_{k=1}^K & \psi^{*-1}\left(\sum_{i=1}^n\frac{I(W_k; Z_{i,k})}{n}\right) \\
& \leq \frac{1}{K^2}\sum_{k=1}^K \psi^{*-1}\left(\frac{I(W_k; S_k)}{n}\right)
\end{align*}
as desired.
\end{proof}
The result in Theorem \ref{thm:breg} is general enough to apply to many problems of interest. For example, if $F(p) = \|p\|_2^2$, then the Bregman divergence $D_F$ gives the ubiquitous squared $\ell^2$ loss, i.e.,
$$D_F(p,q) = \|p-q\|_2^2 \; .$$
For a comprehensive list of realizable loss functions, the interested reader is referred to \cite{banerjee2005clustering}. Using the above $F$, Theorem \ref{thm:breg} can apply to ordinary least squares regression which we will look at in more detail in Section \ref{sec:sims}. Other regression models such as logistic regression have a loss function that cannot be described with a Bregman divergence without the inclusion of an additional nonlinearity. However, the result in Theorem \ref{thm:breg} is agnostic to the algorithm that each node uses to fit its individual model. In this way, each node could be fitting a logistic model to its data, and the total aggregate model would then be an average over these logistic models. Theorem \ref{thm:breg} would still control the expected generalization error for the aggregate model with the extra $\frac{1}{K}$ factor, however, critically, the upper bound would only be for generalization error that is with respect to a loss of the form $D_F(\langle x,w\rangle ,y)$ such as squared $\ell^2$ loss.

In order to show that the dependence on the number of nodes $K$ from Theorem \ref{thm:breg} is tight for certain problems, consider the following example from \cite{veeravalli}. Suppose that $Z\sim\pi=\mathcal{N}(\mu,\sigma^2I_d)$ and $\ell(w,z) = \| w - z \|_2^2$ so that we are trying to learn the mean $\mu$ of the Gaussian. An obvious algorithm for each node to use is simple averaging of its dataset:
$$w_k = \mathcal{A}_k(s_k) = \frac{1}{n}\sum_{i=1}^n z_{i,k} \; .$$ 
For this algorithm, it can be shown that
$$I(\widehat W;Z_{i,k}) = \frac{d}{2}\log\frac{nK}{nK-1}$$ and $$\psi^{*-1}(y) = 2\sqrt{d\left(1+\frac{1}{nK}\right)^2\sigma^4y}$$
(see  Section IV.A. in \cite{veeravalli}). If we apply the existing information theoretic bounds from Theorem \ref{thm:info_theoretic} in an end-to-end way, such as would be the approach from \cite{spawc}, we would get
\begin{align*}
\mathbb{E}_{S\sim\pi^{nK}}[\Delta_\mathcal{A}(S)] & \leq \sigma^2d\sqrt{2\left(1+\frac{1}{nK}\right)^2\log\frac{nK}{nK-1}} \\
& = O\left(\frac{1}{\sqrt{nK}}\right) \; .
\end{align*}
However, for this choice of algorithm at each node, the true expected generalization error can be computed to be
$$\mathbb{E}_{S\sim\pi^{nK}}[\Delta_\mathcal{A}(S)] = \frac{2\sigma^2d}{nK} \; .$$
Applying our new bound from Theorem \ref{thm:breg}, we get
\begin{align*}
\mathbb{E}_{S\sim\pi^{nK}}[\Delta_\mathcal{A}(S)] & \leq \frac{\sigma^2d}{K}\sqrt{2\left(1+\frac{1}{n}\right)^2\log\frac{n}{n-1}}\\
& \leq O\left( \frac{1}{K\sqrt{n}} \right)
\end{align*}
which recovers the correct dependence on $K$ and improves upon the $O\left(\frac{1}{\sqrt{K}}\right)$ result from previous information theoretic methods.

\subsection{General Models and Losses}
In this section we briefly describe some results that hold for more general classes of models and loss functions, such as deep neural networks and other nonlinear models.
\begin{thm}[Lipschitz Continuous Loss] \label{thm:lip}
Suppose that $\ell(w,z)$ is Lipschitz continuous as a function of $w$ in the sense that
$$|\ell(w,z) - \ell(w',z)| \leq C\|w - w'\|_2$$
for any $z$, and that $$\mathbb{E} \left[ \left\| W_k -\mathbb{E}[W_k]\right\|_2 \right] \leq \sigma_0$$ for each $k$.
Then
\begin{align*} 
\mathbb{E}_{S\sim \pi^{nK}} [\Delta_\mathcal{A}(S)] & \leq   \frac{2C\sigma_0}{K}\; .
\end{align*}
\end{thm}
\begin{proof}
Starting with Theorem \ref{thm:leave},
\begin{align} 
\mathbb{E}_{S\sim \pi^{nK}} & [\Delta_\mathcal{A}(S)] \nonumber \\
& = \frac{1}{nK}\sum_{i,k}\mathbb{E}_{S,S'}\left[ \ell(\mathcal{A}(S),Z_{i,k}') - \ell(\mathcal{A}(S^{(i,k)}),Z_{i,k}'\right] \nonumber \\
& \leq  \frac{1}{nK}\sum_{i,k}\mathbb{E}_{S,S'} \left[C \left\| \mathcal{A}(S) - \mathcal{A}(S^{(i,k)})\right\|_2\right]  \label{eq:lip}\\
& = \frac{1}{nK^2}\sum_{i,k}\mathbb{E}_{S,S'} \left[C \left\| W_k - W_k^{(i)}\right\|_2\right] \nonumber\\
& \leq \frac{C}{nK^2}\sum_{i,k}\mathbb{E}_{S,S'} \left[\left\| W_k - \mathbb{E}[W_k]\right\|_2\right] \nonumber  \\
& \quad \quad \quad \quad \quad + \mathbb{E}_{S,S'} \left[\left\| W_k^{(i)} - \mathbb{E}[W_k]\right\|_2\right] \label{eq:lip2} \\
& \leq \frac{2C\sigma_0}{K} \; . \label{eq:lip3}
\end{align}
\noindent Equation \eqref{eq:lip} follows due to Lipschitz continuity, equation \eqref{eq:lip2} uses the triangle inequality, and equation \eqref{eq:lip3} is by assumption. 


\end{proof}
The bound in Theorem \ref{thm:lip} is not in terms of the information theoretic quantities $I(W_k;S_k)$, but it does show that the $O\left(\frac{1}{K}\right)$ upper bound can be shown for much more general loss functions and arbitrary nonlinear models.

\subsection{Privacy and Communication Constraints}

Both communication constraints and local differential privacy constraints can be thought of as special cases of mutual information constraints. Motivated by this observation, Theorem \ref{thm:breg} immediately implies corollaries for these types of system.
\begin{cor}[Privacy Constraints]
Suppose each node's algorithm $\mathcal{A}_k$ is an $\varepsilon$-local differentially private mechanism in the sense that $\frac{p(w_k|s_k)}{p(w_k|s'_k)} \leq e^\varepsilon$ for each $w_k,s_k,s'_k$. Then for losses $\ell$ of the form in Theorem \ref{thm:breg}, and under Assumption \ref{ass2},
\begin{align*}
\mathbb{E}_{S\sim \pi^{nK}} [\Delta_\mathcal{A}(S)] \leq \frac{1}{K}\sqrt{\frac{2R^2\min\{\varepsilon,(e-1)\varepsilon^2\}}{n}} \; .
\end{align*}
\end{cor}
\begin{cor}[Communication Constraints]
Suppose each node can only transit $B$ bits of information to the model aggregator, meaning that each $W_k$ can only take $2^B$ distinct possible values. Then for losses $\ell$ of the form in Theorem \ref{thm:breg}, and under Assumption \ref{ass2},
\begin{align*}
\mathbb{E}_{S\sim \pi^{nK}} [\Delta_\mathcal{A}(S)] \leq \frac{1}{K}\sqrt{\frac{2(\log 2)R^2B}{n}} \; .
\end{align*}
\end{cor}


\section{Iterative Algorithms}

We now turn to considering more complicated multi-round and iterative algorithms. In this setup, after $T$ rounds there is a sequence of weights $W^{(T)} = (W^{1},\ldots,W^{T})$ and the final model $\widehat W_T = f_T(W^{(T)})$ is a function of that sequence where $f_T$ gives a linear combination of the $T$ vectors $W^{1},\ldots,W^{T}$. The function $f_T$ could represent, for example, averaging over the $T$ iterates, picking out the last iterate $W^T$, or some weighted average over the iterates. On each round $t$, each node $k$ produces an updated model $W^t_k$ based on its local dataset $S_k$ and the previous timestep's global model $W^{t-1}$. The global model is then updated via an average over all $K$ updated submodels:
\begin{align*}
W^t = \frac{1}{K}\sum_{k=1}^K W_k^t \; .
\end{align*}
The particular example that we will consider is that of distributed SGD, where each node constructs its updated model $W_k^t$ by taking one or more gradient steps starting from $W^{t-1}$ with respect to random minibatches of its local data. Our model is general enough to account for multiple local gradient steps as is used in so-called Federated Learning \cite{federated0,federated1,federated2}, as well as noisy versions of SGD such as in \cite{jog,wang2021generalization}. If only one local gradient step is taken on each iteration, then the update rule for this particular example could be written as
\begin{align} \label{eq:sgd_update}
W_k^t = W^{t-1} - \eta_t \nabla_w \ell(W^{t-1},Z_{t,k}) + \xi_t
\end{align}
where $Z_{t,k}$ is a data point (or minibatch) sampled from $S_k$ on timestep $t$, $\eta_t$ is the learning rate, and $\xi_t$ is some potential added noise. We assume that the data points $Z_{t,k}$ are sampled without replacement so that the samples are distinct across different values of $t$.

For this type of iterative algorithm, we will consider the following timestep averaged empirical risk quantity:
\begin{align*}
\frac{1}{KT} \sum_{t=1}^T\sum_{k=1}^K \ell(\widehat W_t,Z_{t,k}) \; ,
\end{align*}
and the corresponding generalization error
\begin{align} \label{eq:iter_gen_error}
\Delta_\mathsf{sgd}(S) = \frac{1}{T} \sum_{t=1}^T \left(\mathbb{E}_{Z\sim\pi}[\ell(\widehat W_t,Z)] - \frac{1}{K}\sum_{k=1}^K \ell(\widehat W_t,Z_{t,k})\right) \; .
\end{align}
Note that the quantity in \eqref{eq:iter_gen_error} is slightly different than the end-to-end generalization error that we would get considering the final model $\widehat W_T$ and whole dataset $S$. It is instead an average over the generalization error we would get from each model stopping at iteration $t$. We do this so that when we apply the leave-one-out expansion from Theorem \ref{thm:leave}, we do not have to account for the dependence of $W_k^t$ on past samples $Z_{t',k'}$ for $t'<t$ and $k'\neq k$. Since we expect the generalization error to decrease as we use more samples, this quantity should result in a more conservative upper bound and be a reasonable surrogate object to study. The following bound follows as a corollary to Theorem \ref{thm:breg}.
\begin{cor}
For losses $\ell$ of the form in Theorem \ref{thm:breg}, and under Assumption \ref{ass2},
\begin{align*}
\mathbb{E}\left[\Delta_\mathsf{sgd}(S)\right] \leq \frac{1}{T} \sum_{t=1}^T \frac{1}{K^2} \sum_{k=1}^K \sqrt{2R^2I(W_k^t;Z_{t,k})} \; .
\end{align*}
\end{cor}

\section{Simulations} \label{sec:sims}
We simulated a distributed linear regression example in order to demonstrate the improvement in our bounds over the existing information theoretic bounds. To do this, we generated $n=10$ synthetic datapoints at each of $K$ different nodes for various values of $K$. Each datapoint consisted of a pair $(x,y)$ where $y=xw_0+n$ with $x,n \sim \mathcal{N}(0,1)$, and $w_0\sim \mathcal{N}(0,1)$ was the randomly generated true weight that was common to all datapoints. Each node constructed an estimate $\widehat w_k$ of $w_0$ using the well-known normal equations which minimize the  $\ell^2$ loss, i.e., $\widehat w_k = \argmin_{w} \sum_{i=1}^n (wx_{i,k} - y_{i,k})^2$. The aggregate model was then the average $\widehat w = \frac{1}{K}\sum_{k=1}^K \widehat w_k$. In order to estimate the old and new information theoretic generalization bounds (i.e., the bounds from Theorems \ref{thm:info_theoretic} and \ref{thm:breg}, respectively), this procedure was repeated $M=10^6$ times and the datapoint and model values were binned in order to estimate the mutual information quantities. The value of $M$ was increased until the mutual information estimates were no longer particularly sensitive to the number and widths of the bins. In order to estimate the true generalization error, the expectations for both the population risk and the dataset were estimated by Monte Carlo with $10^4$ trials each. The results can be seen in Figure \ref{fig2}, where it is evident that the new information theoretic bound is much closer to the true expected generalization error, and decays with an improved rate as a function of $K$.
\begin{figure}
\begin{center}
\includegraphics[width=3in]{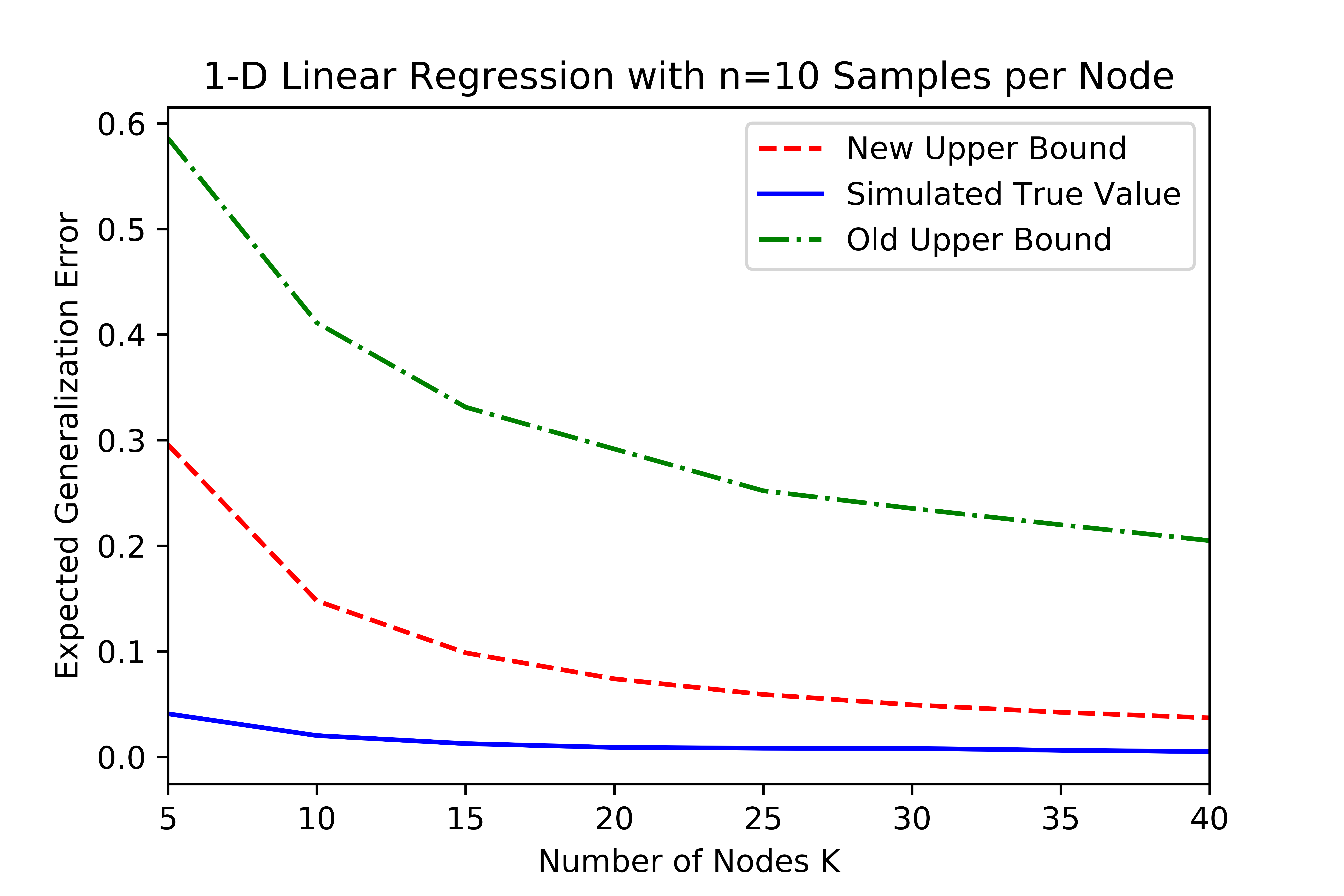}
\includegraphics[width=3in]{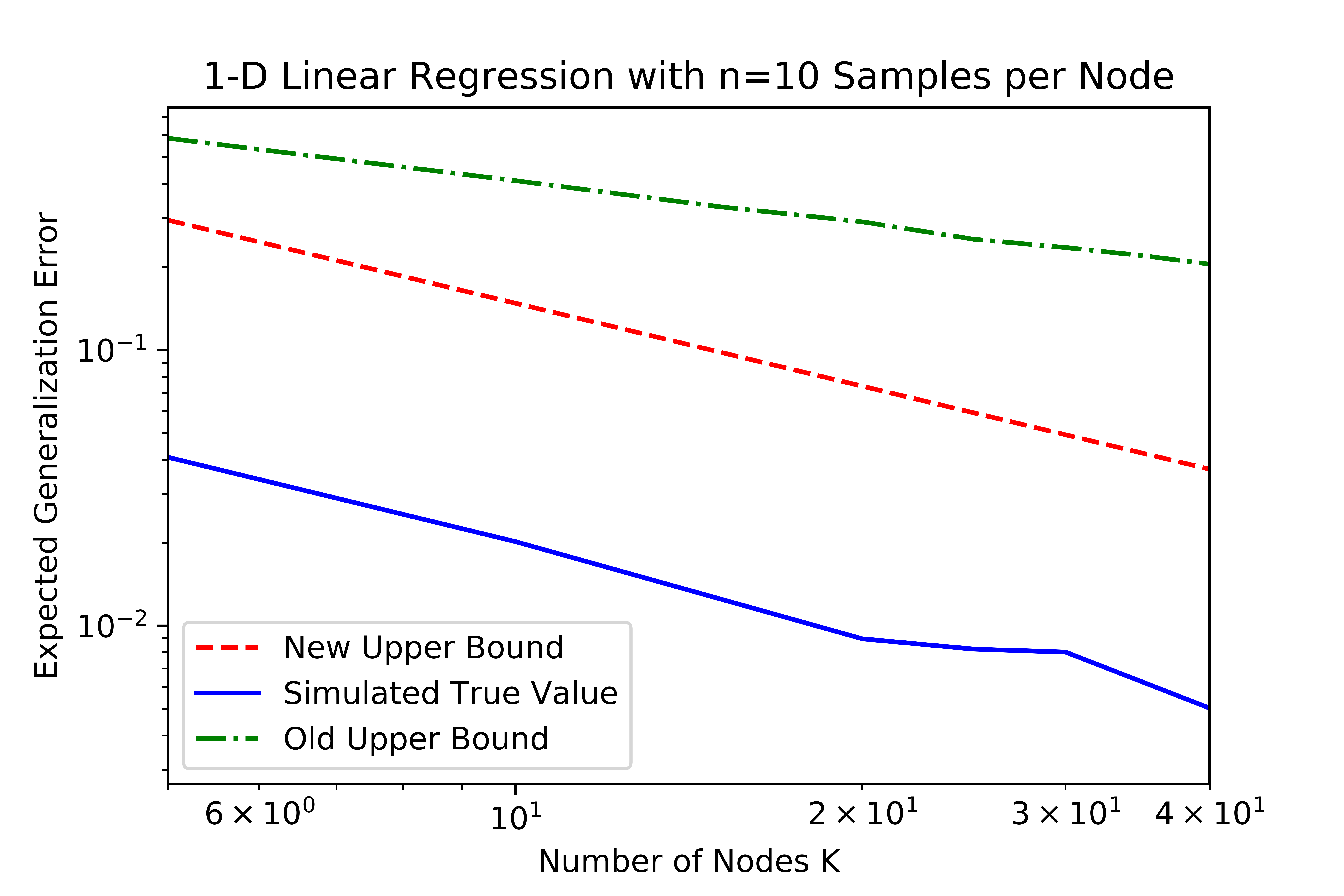}
\end{center}
\caption{Information theoretic upper bounds and expected generalization error for a simulated linear regression example in linear (top) and log (bottom) scales. \label{fig2}}
\end{figure}

\nocite{*}
\bibliographystyle{IEEEtran}
\bibliography{main_v5.bib}

\begin{thebibliography}{10}
\providecommand{\url}[1]{#1}
\csname url@samestyle\endcsname
\providecommand{\newblock}{\relax}
\providecommand{\bibinfo}[2]{#2}
\providecommand{\BIBentrySTDinterwordspacing}{\spaceskip=0pt\relax}
\providecommand{\BIBentryALTinterwordstretchfactor}{4}
\providecommand{\BIBentryALTinterwordspacing}{\spaceskip=\fontdimen2\font plus
\BIBentryALTinterwordstretchfactor\fontdimen3\font minus
  \fontdimen4\font\relax}
\providecommand{\BIBforeignlanguage}[2]{{%
\expandafter\ifx\csname l@#1\endcsname\relax
\typeout{** WARNING: IEEEtran.bst: No hyphenation pattern has been}%
\typeout{** loaded for the language `#1'. Using the pattern for}%
\typeout{** the default language instead.}%
\else
\language=\csname l@#1\endcsname
\fi
#2}}
\providecommand{\BIBdecl}{\relax}
\BIBdecl

\bibitem{russo_zou}
D.~Russo and J.~Zou, ``How much does your data exploration overfit? controlling
  bias via information usage,'' \emph{IEEE Transactions on Information Theory},
  vol.~66, no.~1, pp. 302--323, 2020.

\bibitem{xu_raginsky}
A.~Xu and M.~Raginsky, ``Information-theoretic analysis of generalization
  capability of learning algorithms,'' in \emph{NIPS}, 2017, pp. 2521--2530.

\bibitem{veeravalli}
Y.~Bu, S.~Zou, and V.~V. Veeravalli, ``Tightening mutual information-based
  bounds on generalization error,'' \emph{IEEE Journal on Selected Areas in
  Information Theory}, vol.~1, no.~1, pp. 121--130, 2020.

\bibitem{federated0}
H.~B. McMahan, E.~Moore, D.~Ramage, S.~Hampson, and B.~A. y~Arcas,
  ``Communication-efficient learning of deep networks from decentralized
  data,'' in \emph{Proceedings of AISTATS}, 2017.

\bibitem{federated1}
J.~Konecn{\'y}, H.~B. McMahan, D.~Ramage, and P.~Richt{\'a}rik, ``Federated
  optimization: Distributed machine learning for on-device intelligence,''
  \emph{CoRR}, vol. abs/1610.02527, 2016.

\bibitem{federated2}
J.~Konečný, H.~B. McMahan, F.~X. Yu, P.~Richtarik, A.~T. Suresh, and
  D.~Bacon, ``Federated learning: Strategies for improving communication
  efficiency,'' in \emph{Proceedings of the NIPS Workshop on Private
  Multi-Party Machine Learning}, 2016.

\bibitem{deepgrad}
Y.~Lin, S.~Han, H.~Mao, Y.~Wang, and W.~J. Dally, ``Deep gradient compression:
  Reducing the communication bandwidth for distributed training,''
  \emph{Proceedings of the 6th International Congress on Learning
  Representations (ICLR)}, 2018.

\bibitem{rtopk}
L.~P. Barnes, H.~A. Inan, B.~Isik, and A.~Ozgur, ``r{Top}-k: A statistical
  estimation approach to distributed {SGD},'' \emph{IEEE Journal on Selected
  Areas in Information Theory}, vol.~1, no.~3, pp. 897--907, 2020.

\bibitem{warner}
S.~L. Warner, ``Randomized response: A survey technique for eliminating evasive
  answer bias,'' \emph{Journal of the American Statistical Association},
  vol.~60, no. 309, pp. 63--69, 1965.

\bibitem{dwork1}
C.~Dwork, F.~McSherry, K.~Nissim, and A.~Smith, ``Calibrating noise to
  sensitivity in private data analysis,'' in \emph{Theory of Cryptography
  Conference}, S.~Halevi and T.~Rabin, Eds.\hskip 1em plus 0.5em minus
  0.4em\relax Springer, Berlin, Heidelberg, 2006.

\bibitem{whatcanwelearn}
S.~P. Kasiviswanathan, H.~K. Lee, K.~Nissim, S.~Raskhodnikova, and A.~Smith,
  ``What can we learn privately?'' \emph{SIAM Journal on Computing}, vol.~40,
  no.~3, p. 793–826, 2011.

\bibitem{cuff}
P.~Cuff and L.~Yu, ``Differential privacy as a mutual information constraint,''
  in \emph{Proceedings of the 2016 ACM SIGSAC Conference on Computer and
  Communications Security}, 2016, pp. 43--54.

\bibitem{spawc}
S.~Yagli, A.~Dytso, and H.~Vincent~Poor, ``Information-theoretic bounds on the
  generalization error and privacy leakage in federated learning,'' in
  \emph{Proceedings of the 2020 IEEE 21st International Workshop on Signal
  Processing Advances in Wireless Communications (SPAWC)}, 2020, pp. 1--5.

\bibitem{SSSS}
S.~Shalev-Shwartz, O.~Shamir, N.~Srebro, and K.~Sridharan, ``Learnability,
  stability and uniform convergence,'' \emph{Journal of Machine Learning
  Research}, vol.~11, pp. 2635--2670, 2010.

\bibitem{bregman1967relaxation}
L.~M. Bregman, ``The relaxation method of finding the common point of convex
  sets and its application to the solution of problems in convex programming,''
  \emph{USSR {C}omputational {M}athematics and {M}athematical {P}hysics},
  vol.~7, no.~3, pp. 200--217, 1967.

\bibitem{banerjee2005clustering}
A.~Banerjee, S.~Merugu, I.~S. Dhillon, J.~Ghosh, and J.~Lafferty, ``Clustering
  with {B}regman divergences.'' \emph{Journal of Machine Learning Research},
  vol.~6, no.~10, 2005.

\bibitem{jog}
A.~Pensia, V.~Jog, and P.-L. Loh, ``Generalization error bounds for noisy,
  iterative algorithms,'' in \emph{Proceedings of the 2018 IEEE International
  Symposium on Information Theory (ISIT)}, 2018, pp. 546--550.

\bibitem{wang2021generalization}
H.~Wang, R.~Gao, and F.~P. Calmon, ``Generalization bounds for noisy iterative
  algorithms using properties of additive noise channels,'' 2021.

\bibitem{vondrak}
\BIBentryALTinterwordspacing
V.~Feldman and J.~Vondrak, ``High probability generalization bounds for
  uniformly stable algorithms with nearly optimal rate,'' in \emph{Proceedings
  of the Thirty-Second Conference on Learning Theory}, ser. Proceedings of
  Machine Learning Research, A.~Beygelzimer and D.~Hsu, Eds., vol.~99.\hskip
  1em plus 0.5em minus 0.4em\relax Phoenix, USA: PMLR, 25--28 Jun 2019, pp.
  1270--1279. [Online]. Available:
  \url{http://proceedings.mlr.press/v99/feldman19a.html}
\BIBentrySTDinterwordspacing

\bibitem{gastpar}
A.~R. Esposito, M.~Gastpar, and I.~Issa, ``Generalization error bounds via
  r\'enyi-$f$-divergences and maximal leakage,'' \emph{IEEE Transactions on
  Information Theory}, vol.~67, no.~8, pp. 4986--5004, 2021.

\bibitem{verdu}
A.~R. Asadi, E.~Abbe, and S.~Verd\'{u}, ``Chaining mutual information and
  tightening generalization bounds,'' in \emph{Proceedings of the 32nd
  International Conference on Neural Information Processing Systems}, ser.
  NIPS'18, 2018, p. 7245–7254.

\bibitem{rakhlin}
M.~Raginsky, A.~Rakhlin, M.~Tsao, Y.~Wu, and A.~Xu, ``Information-theoretic
  analysis of stability and bias of learning algorithms,'' in \emph{Proceedings
  of the 2016 IEEE Information Theory Workshop (ITW)}, 2016, pp. 26--30.

\bibitem{yanjun}
J.~Jiao, Y.~Han, and T.~Weissman, ``Dependence measures bounding the
  exploration bias for general measurements,'' in \emph{Proceedings of the 2017
  IEEE International Symposium on Information Theory (ISIT)}.\hskip 1em plus
  0.5em minus 0.4em\relax IEEE, 2017, pp. 1475--1479.

\end{thebibliography}

\end{document}